\documentclass{amsart}
\usepackage{graphicx}
\newtheorem{thm}{Theorem}[section]
\newtheorem{cor}[thm]{Corollary}

\theoremstyle{definition}

\theoremstyle{remark}
\newtheorem{rem}[thm]{Remark}
\numberwithin{equation}{section}

\begin{document}

\title[Universal Monotonicity of Eigenvalue Moments...]{Universal Monotonicity of Eigenvalue Moments and Sharp Lieb-Thirring Inequalities}%
\author{Joachim Stubbe}%
\address{Joachim Stubbe, EPFL, IMB-FSB, Station 8, CH-1015 Lausanne, Switzerland}%
\email{Joachim.Stubbe@epfl.ch}%
\thanks{The author would like to thank Evans M. Harrell II for fruitful discussions and many helpful comments.
Valuable suggestions by R. Frank and E.H. Lieb are also gratefully acknowledged}%

\date{13 October 2008}%
\begin{abstract}
We show that phase space bounds on the eigenvalues of
Schr\"{o}dinger operators can be derived from universal bounds
recently obtained by E. M. Harrell and the author via a
monotonicity property with respect to coupling constants. In
particular, we provide a new proof of sharp Lieb-Thirring
inequalities.
\end{abstract}
\maketitle
\section{Introduction}
We consider the eigenvalues $E_j(\alpha)$ of a one-parameter
family of Schr\"{o}dinger operators
\begin{equation}\label{H-alpha}
    H(\alpha)=-\alpha\Delta+V(x)
\end{equation}
on $\mathbb{R}^d$ for constants $\alpha>0$. For negative
potentials $V(x)$ vanishing at infinity it is a well-known fact
that for all $\sigma\geq 0$
\begin{equation}\label{sc-limit}
    \underset{\alpha\rightarrow 0+}{\lim}\alpha^{\frac{d}{2}}\;
    \sum_{E_j(\alpha)<0}(-E_j(\alpha))^{\sigma}=L_{\sigma,d}^{cl}\int_{\mathbb{R}^d}(-V(x))^{\sigma+\frac{d}{2}}\;dx
\end{equation}
with $L_{\sigma,d}^{cl}$, called the classical constant, given by
\begin{equation}\label{sc-constant}
    L_{\sigma,d}^{cl}=(4\pi)^{-\frac{d}{2}}\frac{\Gamma(\sigma+1)}{\Gamma(\sigma+\frac{d}{2}+1)}.
\end{equation}
Lieb-Thirring inequalities are inequalities of the form
\begin{equation}\label{LT-ineq}
    \alpha^{\frac{d}{2}}\;
    \sum_{E_j(\alpha)<0}(-E_j(\alpha))^{\sigma}\leq L_{\sigma,d}\int_{\mathbb{R}^d}(-V(x))^{\sigma+\frac{d}{2}}\;dx
\end{equation}
for some constant $L_{\sigma,d}\geq L_{\sigma,d}^{cl}$ and are
widely discussed in the literature (see e.g.
\cite{BlSt,Hu,LaWei2}). A longstanding question is when
\eqref{LT-ineq} holds with $L_{\sigma,d}= L_{\sigma,d}^{cl}$. The
most general result is due to Laptev and Weidl \cite{LaWei1} who
proved that $L_{\sigma,d}= L_{\sigma,d}^{cl}$ for all $\sigma\geq
\frac{3}{2}$ and $d\geq 1$. Their proof is based on a dimensional
reduction of Schr\"{o}dinger operators with operator valued
potentials which allows them to make use of the bound for
$\sigma=\frac{3}{2},d=1$ which has been first proven by Lieb and
Thirring \cite{LT1}. For simplified proof see also \cite{BeLo}. On
the other hand, by analyzing the spectra of harmonic oscillators
Helffer and Robert have shown that $L_{\sigma,d}>
L_{\sigma,d}^{cl}$ for $\sigma<1$ while de la Breteche showed that
these spectra are in agreement with the conjecture $L_{\sigma,d}=
L_{\sigma,d}^{cl}$ for $\sigma\geq 1$ \cite{dlB}

Recently, Harrell and the author have established universal trace
inequalities for abstract self-adjoint operators $H$ modelled on
Schr\"{o}dinger operators \cite{HaSt2}. If $G$ is another
self-adjoint operator, then under suitable domain conditions (see
Corollary 2.3 of \cite{HaSt2})
   \begin{align}\label{HS-quadratic}
    & \sum_{E_j\in J}{(z-E_j)^2\,\langle[G, [H,G]]\phi_j,\phi_j\rangle -
 2(z-E_j)\,\langle[H,G]\phi_j,[H,G]\phi_j\rangle}
    \\
    &\qquad =2 \sum_{E_j\in J}\int_{\kappa \in J^c}
    (z-E_j)(z-\kappa)(\kappa-E_j) d G_{j \kappa}^2\nonumber
\end{align}
where $J$ denotes a subset of the discrete spectrum of $H$ and
$J^c$ its complement. Exploiting this identity we prove the
following
\begin{thm}\label{main-theorem} Suppose $V\leq 0$ vanishes at infinity. Then the
mapping
\begin{equation}\label{diff-ineq-2}
   \alpha\mapsto\alpha^{\frac{d}{2}}\;
    \sum_{E_j(\alpha)<0}(-E_j(\alpha))^{2}
\end{equation}
is non increasing for all $\alpha>0$. Consequently
\begin{equation}\label{LT-2-sharp}
     \alpha^{\frac{d}{2}}\;
    \sum_{E_j(\alpha)<0}(-E_j(\alpha))^{2}\leq
    L_{2,d}^{cl}\int_{\mathbb{R}^d}(-V(x))^{2+\frac{d}{2}}\;dx
\end{equation}
for all $\alpha>0$.
\end{thm}

The link between universal inequalities and semiclassical
estimates has been first made in \cite{HaSt} where it has been
shown for the Dirichlet Laplacian $-\Delta_D$ on a bounded domain
$D\in\mathbb{R}^d$ that the mapping
\begin{equation}\label{Dirichlet-heat-kernel-trace}
   t\mapsto t^{\frac{d}{2}}\;\text{tr}(e^{-t\Delta_D})
\end{equation}
is always decreasing and therefore bounded by its semiclassical
limit, that is
\begin{equation}\label{Dirichlet-heat-kernel-trace-estimate}
    \text{tr}(e^{-t\Delta_D})\leq (4\pi t)^{-\frac{d}{2}} |D|.
\end{equation}
Harrell and Hermi have extended this technique to Riesz means of
the Dirichlet Laplacian \cite{HaHe}. In \cite{HaSt2} it has been
pointed out that the monotonicity of mappings like
\eqref{Dirichlet-heat-kernel-trace} is a universal property of a
large family of "trace-controllable" functions (as precisely
defined in \cite{HaSt2}) of Schr\"{o}dinger operators and we shall
derive in the present paper a corresponding universal property of
one parameter families of Schr\"{o}dinger operators. Our second
result extends this property to Schr\"{o}dinger operators of the
form \eqref{H-alpha} with confining potentials $V(x)$ such that
\begin{equation}\label{exp-V}
    \int_{\mathbb{R}^d}e^{-tV(x)}\;dx<\infty
\end{equation}
for all $t>0$. We provide a monotonicity result implying the
Golden-Thompson inequality for Schr\"{o}dinger operators
\eqref{H-alpha}:

\begin{thm}\label{main-theorem-2} If \eqref{exp-V} holds, then for all
$t>0$ the mapping
\begin{equation}\label{diff-ineq-3}
     \alpha\mapsto\alpha^{\frac{d}{2}}\;\text{tr}\,(e^{-tH(\alpha)})
\end{equation}
is non increasing for all $\alpha>0$. Consequently, for all
$\alpha>0$,
\begin{equation}\label{GT-sharp}
    \text{tr}\,(e^{-tH(\alpha)})\leq (4\pi \alpha t)^{-\frac{d}{2}}\int_{\mathbb{R}^d}e^{-tV(x)}\;dx<\infty.
\end{equation}
\end{thm}

\section{Proof of Main Results}
The key for proving our main results is the trace formula for self
adjoint operators proved in \cite{HaSt2}. For convenience we
reformulate this result for the operator $H(\alpha)$ in a slightly
different, and as we believe, in a more transparent way. To make
the present paper self consistent we give an elementary proof of
the trace formula. For simplicity, we consider only the case of
purely discrete spectra (more relevant for Theorem
\ref{main-theorem-2}). In the presence of continuous spectrum one
uses the spectral integral as in \cite{HaSt2}.
\begin{thm}[Trace formula for $H(\alpha)$]\label{HS-trace-formula} Suppose that
$H(\alpha)$ given in \eqref{H-alpha} has a spectrum consisting of
eigenvalues $E_k=E_k(\alpha)$ with associated eigenfunctions
$\phi_k$ forming an orthonormal basis of the underlying Hilbert
space $L^2(\mathbb{R}^d)$. Let $f:\mathbb{R}\rightarrow
\mathbb{R}$ be a $C^1$-function. Then
\begin{equation}\label{HS-trace}
    d\sum_{E_j}f(E_j)+2\alpha\sum_{E_j}\sum_{E_k}T_{jk}\int_0^1f'(sE_j+(1-s)E_k)\;ds=0
\end{equation}
provided all sums are finite where
\begin{equation}\label{T-jk}
    T_{jk}=T_{kj}=\bigg|\int_{\mathbb{R}^d}\phi_j\nabla\overline{\phi}_k\;dx\bigg|^2
\end{equation}
denote the kinetic energy matrix elements.
\end{thm}
\begin{proof}
Let $x_a$, $a=1,\ldots,d$ denote cartesian coordinates in
$\mathbb{R}^d$ and $D_a=\partial/\partial x_a$. The first identity
is due to canonical commutation (or integration by parts) and the
completeness of eigenfunctions. Indeed, for all $j$
\begin{equation}\label{canon-comm}
\begin{split}
    1&=-\int_{\mathbb{R}^d}x_a\phi_jD_a\overline{\phi}_j\;dx-\int_{\mathbb{R}^d}x_a\overline{\phi}_jD_a{\phi}_j\;dx\\
    &=-\sum_{k}\int_{\mathbb{R}^d}x_a\phi_j\overline{\phi}_k\;dx\int_{\mathbb{R}^d}\phi_kD_a\overline{\phi}_j\;dx
    +\int_{\mathbb{R}^d}x_a\overline{\phi}_j{\phi}_k\;dx\int_{\mathbb{R}^d}\overline{\phi}_kD_a{\phi}_j\;dx\\
    \end{split}
\end{equation}
Next we apply the gap formula
\begin{equation}\label{HS-trace-bis}
    (E_k-E_j)\int_{\mathbb{R}^d}x_a\phi_j\overline{\phi}_k\;dx=-2\alpha
    \int_{\mathbb{R}^d}(D_a\phi_j)\overline{\phi}_k\;dx.
\end{equation}
We note that the r.h.s is zero for degenerate eigenvalues.
Therefore after summing over all coordinates in \eqref{canon-comm}
we get
\begin{equation}\label{canon-comm-2}
    d=4\alpha\sum_{E_k}\frac{T_{jk}}{E_k-E_j}.
\end{equation}
Multiplying \eqref{canon-comm-2} by $f(E_j)$, summing over $j$ and
symmetrizing the double sum we finally obtain
\begin{equation*}
    d\sum_{E_j}f(E_j)+2\alpha\sum_{E_j}\sum_{E_k}T_{jk}\frac{f(E_k)-f(E_j)}{E_k-E_j}=0.
\end{equation*}
The assertion follows then by the fundamental theorem of calculus.
\end{proof}
Applying theorem \ref{HS-trace-formula} to $f(E)=(z-E)^2$ for
$E<z$ and $f(E)=0$ otherwise we recover \eqref{HS-quadratic} with
$G$ being the multiplication operator $x_a$ after summing over all
coordinates as shown in \cite{HaSt, HaSt2}:
\begin{equation}\label{HS-quadratic-bis}
    \sum_{E_j<z}d(z-E_j)^2-4\alpha(z-E_j)T_j
    =4 \sum_{E_j<z}\sum_{E_k\geq z}T_{jk}
    \frac{(z-E_j)(z-E_k)}{E_k-E_j}
\end{equation}
with
\begin{equation*}
    T_j=\sum_{E_k}T_{jk}=\int_{\mathbb{R}^d}|\nabla\phi_j|^2\;dx.
\end{equation*}
\begin{rem} Formula \eqref{canon-comm-2} can also be easily
derived from second order perturbation theory. Indeed, for a fixed
vector $v\in\mathbb{R}^d$ consider the operator
$H=(-i\sqrt{\alpha}\,\nabla+\gamma v)^2+V(x)$. Obviously, the
addition of a constant vector field does not change the
eigenvalues and second order perturbation (i.e first order in
$\gamma^2v^2$ and second order in
$-2i\sqrt{\alpha}\,\gamma\,v\nabla$ yields \eqref{canon-comm-2}
when choosing $v$ to be the canonical unit vectors $e_a$ and then
summing over all $a=1,\ldots d$. The author thanks R. Seiringer
for indicating this proof.
\end{rem}
Choosing $f$ appropriately in theorem \ref{HS-trace-formula} we
may now prove our main results.
\begin{proof}[Proof of Theorem \ref{main-theorem}]
Without loss of generality we may suppose that $V\in
C_0^{\infty}$. Hence for any $\alpha>0$ the operator $H(\alpha)$
has at most a finite number of eigenvalues. Obviously, the r.h.s
in \eqref{HS-quadratic-bis} is negative. Making the dependence on
the parameter $\alpha$ explicit we have therefore for all $z\leq
0$ the inequality
\begin{equation}\label{HS-quadratic-ineq}
    \alpha\sum_{E_j(\alpha)<0}(z-E_j(\alpha))^2-\frac{4}{d}\;\alpha^2\sum_{E_j(\alpha)<0}(z-E_j(\alpha))T_j(\alpha) \leq
    0.
\end{equation}
The functions $E_j(\alpha)$ are non positive, continuous and
increasing. Furthermore, let $\infty>\alpha_1\geq
\alpha_2\geq\ldots\geq \alpha_k\geq\ldots >0$ denote the values at
which $E_j(\alpha)$ appears. $E_j(\alpha)$ is continuously
differentiable for $\alpha\neq \alpha_k$ and by the
Feynman-Hellmann theorem
\begin{equation}\label{FH}
    \frac{d}{d\alpha}\;E_j(\alpha)=T_j(\alpha).
\end{equation}
Taking $z=0$, inequality\eqref{HS-quadratic-ineq} then reads
\begin{equation*}
    \alpha\sum_{E_j(\alpha)<0}(-E_j(\alpha))^2+\frac{2}{d}\;\alpha^2\frac{d}{d\alpha}\sum_{E_j(\alpha)<0}(-E_j(\alpha))^2 \leq
    0.
\end{equation*}
For any $\alpha\in\;]\alpha_{N+1},\alpha_{N}[$ the number of
eigenvalues is constant and therefore
\begin{equation*}
     \frac{d}{d\alpha}\bigg(\alpha^{\frac{d}{2}}\;
    \sum_{E_j(\alpha)<0}(-E_j(\alpha))^{2}\bigg)\leq 0
\end{equation*}
proving the theorem.
\end{proof}
\begin{rem} Strictly speaking the Feynman-Hellmann theorem only
holds for nondegenerate eigenvalues. In the case of degenerate
eigenvalues one has to take the right basis in the corresponding
eigenspace and to change the numbering if need be (see e.g.
\cite{Th}).

\end{rem}
\begin{proof}[Proof of Theorem \ref{main-theorem-2}]
Choose $f(E)=e^{-tE}$ and $t>0$. Since $f'(E)=-tf(E)$ is concave
it follows that
\begin{equation*}
    f'(sE_j+(1-s)E_k)\geq sf'(E_j)+(1-s)f'(E_k).
\end{equation*}
Using the symmetry of $T_{jk}$ we get

\begin{equation*}
    d\sum_{E_j(\alpha)}f(E_j(\alpha))+2\alpha
    \sum_{E_j(\alpha)}f'(E_j(\alpha))T_j(\alpha)\leq 0
\end{equation*}
and we conclude as in the proof of \ref{main-theorem}.
\end{proof}
\section{Extensions and Discussion}
It has already been shown in \cite{HaSt, HaHe, HaSt2} that one can
obtain trace inequalities for the functions $f(E)=(z-E)^{\sigma}$
with $\sigma\geq 2$. In fact we have the following result:
\begin{cor}
Let $f:\mathbb{R}\rightarrow \mathbb{R}$ be a $C^1$ function with
support on the negative half axis such that $f'$ is concave. Under
the conditions of theorem \ref{main-theorem} the mapping
\begin{equation}\label{H-alpha-general-f}
 \alpha\mapsto   \alpha^{\frac{d}{2}}\;
    \sum_{E_j(\alpha)<0}f(E_j(\alpha))
\end{equation}
is non increasing for $\alpha>0$. In particular, for all
$\sigma\geq 2$
\begin{equation}\label{diff-ineq-sigma}
     \alpha\mapsto \alpha^{\frac{d}{2}}\;
    \sum_{E_j(\alpha)<0}(-E_j(\alpha))^{\sigma}
\end{equation}
is non increasing for $\alpha>0$. Consequently,
\begin{equation}\label{LT-sigma-sharp}
     \alpha^{\frac{d}{2}}\;
    \sum_{E_j(\alpha)<0}(-E_j(\alpha))^{\sigma}\leq
    L_{\sigma,d}^{cl}\int_{\mathbb{R}^d}(-V(x))^{\sigma+\frac{d}{2}}\;dx.
\end{equation}
\end{cor}
\begin{proof}
Consider the trace formula \eqref{HS-trace}. The concavity of $f'$
implies that
\begin{equation*}
    \int_0^1f'(sE_j+(1-s)E_k)\;ds\geq \frac{1}{2}f'(E_j)+\frac{1}{2}f'(E_k)
\end{equation*}
Using the symmetry of $T_jk$ and  $f'(E)=0$ for $E\geq 0$ we get
\begin{equation*}
    d\sum_{E_j(\alpha)<0}f(E_j(\alpha))+2\alpha
    \sum_{E_j(\alpha)<0}f'(E_j(\alpha))T_j(\alpha)\leq 0
\end{equation*}
As in the proof of theorem \ref{main-theorem} we use the
Feynman-Hellmann theorem to prove \eqref{H-alpha-general-f}.
\end{proof}
\begin{rem}The sharp Lieb-Thirring inequality
\eqref{LT-sigma-sharp} follows also from \eqref{LT-2-sharp} of
theorem \ref{main-theorem} via the Aizenman-Lieb monotonicity
principle \cite{AL}. However, the monotonicity of the mapping
\eqref{diff-ineq-sigma} is a stronger and new result.
\end{rem}
\begin{rem}
We cannot expect that the monotonicity holds for moments with
$\sigma\leq 2$. For example, consider the $d$-dimensional harmonic
oscillator with eigenvalues
$E_j(\alpha)=\sqrt{\alpha}\;(2j_1+\ldots 2j_d+d)$ for natural
numbers $j_1\ldots 2j_d$. We want to study behavior of the
eigenvalue moments
\begin{equation*}
    S_{\sigma}(\alpha)=\sum_{E_j(\alpha)<1}(1-E_j(\alpha))^{\sigma}.
\end{equation*}
Then for all $\alpha\in[(d+2)^{-2},d^{-2}]$ we have
\begin{equation*}
    \alpha^{\frac{d}{2}}S_{\sigma}(\alpha)= \alpha^{\frac{d}{2}}(1-d\sqrt{\alpha})^{\sigma}.
\end{equation*}
It is easy to see that its  derivative (w.r.t. $\alpha$) is
strictly positive at $\alpha=(d+2)^{-2}$ for all $0\leq\sigma< 2$.
This behavior persists also for the sum of the first two
eigenvalues. Indeed, for all $\alpha\in[(d+4)^{-2},(d+2)^{-2}]$ we
have (taking into account the multiplicity of the second
eigenvalue)
\begin{equation*}
    S_{\sigma}(\alpha)=(1-d\sqrt{\alpha})^{\sigma}+d(1-(d+2)\sqrt{\alpha})^{\sigma}.
\end{equation*}
Then the function
$p_{\sigma}(\alpha):=\alpha^{\frac{d}{2}}S_{\sigma}(\alpha)$ has a
strictly positive derivative at $\alpha=(d+4)^{-2}$ for all
$0\leq\sigma<2$. Obviously, for $\sigma=2$ the derivatives at
these points vanish.
\end{rem}

\bibliographystyle{amsplain}
\bibliography{references}

\end{document}